\documentclass[11pt]{article}

\usepackage{amsmath, amsthm, amssymb, graphicx, enumerate, fullpage, url}
\usepackage{color, hyperref}
\usepackage{algorithm, algpseudocode}

\hypersetup{
	colorlinks=true,
	linkcolor=blue,
	citecolor=magenta
}

\theoremstyle{plain}
\newtheorem{theorem}{Theorem}[section]
\newtheorem{proposition}[theorem]{Proposition}
\newtheorem{lemma}[theorem]{Lemma}

\newtheorem{claim}[theorem]{Claim}

\theoremstyle{definition}

\newcommand{\Z}{\mathbb{Z}}

\newcommand{\cN}{\mathcal{N}}
\newcommand{\cP}{\mathcal{P}}

\DeclareMathOperator{\agr}{agr}
\DeclareMathOperator{\aHom}{aHom}
\DeclareMathOperator{\Alt}{Alt}
\DeclareMathOperator{\Aut}{Aut}

\DeclareMathOperator{\Eq}{Eq}
\DeclareMathOperator{\Hom}{Hom}
\DeclareMathOperator{\id}{id}
\DeclareMathOperator{\im}{im}

\DeclareMathOperator{\Perm}{Perm}

\begin{document}

\title{Group homomorphisms as error correcting codes}
\author{Alan Guo
\thanks{CSAIL, Massachusetts Institute of
Technology, 32 Vassar Street, Cambridge, MA, USA. {\tt aguo@mit.edu}. Research
supported in part by NSF grants CCF-0829672, CCF-1065125,
and CCF-6922462, and an NSF Graduate Research Fellowship}
}
\maketitle

\begin{abstract}
We investigate the minimum distance of the error correcting code formed by the
homomorphisms between two finite groups $G$ and $H$.
We prove some general structural
results on how the distance behaves with respect to natural group operations,
such as passing to subgroups and quotients, and taking products.
Our main result is a general formula for the distance when $G$ is
solvable or $H$ is nilpotent, in terms of the normal subgroup
structure of $G$ as well as the prime divisors of $|G|$ and $|H|$.
In particular, we show that in the above case, the distance is
independent of the subgroup structure of $H$.
We complement this by showing that, in general, the distance depends on the
subgroup structure $G$.

\end{abstract}

\section{Introduction}
\label{section:introduction}

\subsection{Error correcting codes}
The theory of error correcting codes studies \emph{codes}, which are subsets
of $\Sigma^n$ for some alphabet $\Sigma$ and block length $n$.
The distance between two strings of equal length is the number of
coordinates in which they differ. The distance $\Delta$ of a code is simply the
minimum distance between any pair of distinct codewords (elements of the code).
Hamming~\cite{hamming} identifies the distance of a code as the
key parameter measuring the error correcting capability of the code.
As long as the number of coordinates in which a codeword is corrupted is less
than $\Delta/2$, one can uniquely recover the original
codeword.
Elias~\cite{elias} and Wozencraft~\cite{wozencraft} proposed
\emph{list decoding}, in which one insists only on recovering a list,
whose size is at most polynomial in $n$, which contains the original codeword.
The Johnson bound~\cite{johnson1} shows that codes can list decode errors
beyond $\Delta/2$.
Codes with efficient list decoding algorithms include the
Hadamard code~\cite{GL89}, Reed-Solomon codes and variants
thereof~\cite{Sud97,GS99,GR08,Gur11}, Reed-Muller codes~\cite{GKZ08,Gop13},
multiplicity/derivative codes~\cite{Kop12,GW11}, and abelian
group homomorphisms~\cite{GKS06,DGKS08}. For some of these codes, in particular
for carefully chosen subcodes
the folded Reed-Solomon codes and multiplicity/derivative codes~\cite{DL12},
the Reed-Muller codes, and
abelian group homomorphisms, it was shown that for any constant $\epsilon > 0$
one can algorithmically list decode up to $\Delta - \epsilon n$ errors with a
constant list size, depending only on $1/\epsilon$.
For all of these codes, the codewords are interpreted as certain functions
$f:A \to B$ from some domain $A$ to codomain $B$. In this case, the coordinates
of the codeword are indexed by $A$ and the alphabet is~$B$.

In a companion work~\cite{GS14}, the author and Sudan show the analogous
list decoding results for group homomorphisms between supersolvable groups.
A technical obstacle which did not arise in the previous works
of~\cite{GKS06,DGKS08} on list decoding abelian group homomorphisms is
actually determining the distance of the code. This turns out to be a
nontrivial problem and serves as the primary motivation of this paper.

\subsection{Group homomorphisms}

Let $G$ and $H$ be finite groups, with homomorphisms $\Hom(G,H)$.
A function $\phi:G \to H$ is a (left) affine homomorphism if there exists
$h \in H$ and $\phi_0 \in \Hom(G,H)$ such that $\phi(g) = h\phi_0(g)$ for
every $g \in G$.
The set of left affine homomorphisms from $G$ to $H$ by $\aHom(G,H)$.
Note that the set of left affine homomorphisms equals the set of right affine
homomorphisms, since
\[
h\phi_0(g) = (h\phi_0(g)h^{-1})h
\]
and $\psi_0(g) \triangleq h\phi_0(g)h^{-1}$ is a homomorphism.

The \emph{equalizer} of two functions $f,g : G \to H$, denoted
$\Eq(f,g)$, is the set
\[
\Eq(f,g) \triangleq \{x \in G \mid f(x) = g(x)\}.
\]
More generally, if $\Phi \subseteq \{f:G \to H\}$ is a collection of functions,
then the \emph{equalizer} of $\Phi$ is the set
\[
\Eq(\Phi) \triangleq \{x \in G \mid f(x) = g(x)~~\forall f,g \in \Phi\}.
\]
In the theory of error correcting codes, the usual measure of distance between
two strings is the relative Hamming distance, which is the fraction of symbols
on which they differ. In the context of group homomorphisms, we find it
more convenient to study the complementary notion, the fractional agreement.
We define the \emph{agreement} $\agr(f,g)$ between two functions
$f,g:G \to H$ to be the quantity
\[
\agr(f,g) \triangleq \frac{|\Eq(f,g)|}{|G|}.
\]
The \emph{maximum agreement} of the code $\aHom(G,H)$, denoted by
$\Lambda_{G,H}$, is defined as
\[
\Lambda_{G,H} \triangleq
\max_{\substack{\phi,\psi \in \aHom(G,H) \\ \phi \ne \psi}}
\agr(\phi,\psi)
\]
In Section~\ref{section:equalizers}, we study the structure of the equalizers
of homomorphisms and prove some basic results that will be useful later.
As we will see (Proposition~\ref{proposition:lambda hom}),
adding affine homomorphisms does not change the distance of this code.
However, we include these functions in the code so that $\Lambda_{G,H}$ is
well-defined when $|\Hom(G,H)|=1$, as long as $H$ is nontrivial.

\subsection{Our results}

Our main result is the following formula for $\Lambda_{G,H}$ when
$G$ is solvable or $H$ is nilpotent.

\begin{theorem}
\label{theorem:main}
Let $G$ and $H$ be finite groups.
Define
\[
\mathcal{P}_{G,H} \triangleq
\{p \mid \text{$p$ is a prime divisor of $\gcd(|G|,|H|)$}\}
\]
and
\[
\mathcal{N}_{G} \triangleq
\{m \mid \text{$G$ has a proper normal subgroup of index $m$}\}.
\]
If $G$ is solvable or $H$ is nilpotent, then
\[
\Lambda_{G,H} =
\begin{cases}
0 & \text{if $\cP_{G,H} \cap \cN_{G} = \emptyset$}, \\
\frac{1}{\min \cP_{G,H} \cap \cN_{G}}
& \text{if $\cP_{G,H} \cap \cN_{G} \ne \emptyset$}.
\end{cases}
\]
\end{theorem}

In Section~\ref{section:distance general facts}, we prove general facts
about $\Lambda_{G,H}$, such as how it behaves with respect to group
decompositions, subgroups, and quotients.

The proof of Theorem~\ref{theorem:main} is divided into two sections.
Section~\ref{section:distance nilpotent range} handles the case where $H$ is
nilpotent, and Section~\ref{section:distance solvable groups} handles the case
where $G$ is solvable.

In Section~\ref{section:distance simple groups}, we investigate
$\Lambda_{G,H}$ when $G$ is a non-abelian simple group, and in particular when
$G = A_n$ is the alternating group on $n \ge 5$ objects.
We show that the formula for $\Lambda_{G,H}$ for solvable $G$ does not apply to
non-abelian simple groups, and hence does not extend to arbitrary groups.
We also see that, in general, $\Lambda_{G,H}$ depends not only on the prime
divisors of $G$ and $H$ but also on the subgroup structure of $H$, in
particular whether $H$ contains isomorphic copies of $G$ and how these copies
are embedded in $H$.

\section{Equalizers}
\label{section:equalizers}

We begin by observing that the equalizer of a set of (affine) homomorphisms is
a (coset of a) subgroup of $G$.

\begin{proposition}
\label{proposition:equalizer is coset}
Let $G$ and $H$ be finite groups. If $\Phi \subseteq \Hom(G,H)$, then
$\Eq(\Phi)$ is a subgroup of $G$. If $\Phi' \subseteq \aHom(G,H)$ and
$\Eq(\Phi') \ne \emptyset$, then
there exists $\Phi \subseteq \Hom(G,H)$ with $|\Phi| = |\Phi'|$ such that
$\Eq(\Phi')$ is a coset of $\Eq(\Phi)$.
\end{proposition}
%

A basic question we would like to answer is the following: if
$\phi,\psi \in \Hom(G,H)$, then must the index of $\Eq(\phi,\psi)$ divide
$|H|$? Note that this is true when one of the homomorphisms, say $\psi$, is the
trivial homomorphism mapping to $1_H$, so that $\Eq(\phi,\psi) = \ker\phi$.
This follows from the fact that $G/\ker\phi \cong \im\phi$ which is a subgroup
of $H$, so $[G:\ker\phi] = |\im\phi|$ divides $H$. We will show in
Proposition~\ref{proposition:index of equalizer to p-group} that the more
general statement holds when $H$ is a $p$-group. Before doing so, we collect a
few more basic facts that will be useful to us.

\begin{proposition}
\label{proposition:nonempty intersection of preimage is coset}
Let $G$ and $H$ be finite groups and let $\Phi \subseteq \Hom(G,H)$.
For $h \in H$, if the set
$\bigcap_{\phi \in \Phi} \phi^{-1}(h)$ is nonempty, then it is
a coset of the subgroup $\bigcap_{\phi \in \Phi}\ker\phi$.
\end{proposition}

\begin{proposition}
\label{proposition:intersection of normal subgroups}
Let $G$ be a group with normal subgroups $N_1,\ldots,N_k \triangleleft G$.
Then $N \triangleq \bigcap_{i=1}^k N_i$ is a normal subgroup of $G$ and
$G/N$ is isomorphic to a subgroup of $\bigoplus_{i=1}^k (G/N_i)$.
\end{proposition}
\begin{proof}
Consider the homomorphism $\phi:G \to \bigoplus_{i=1}^k (G/N_i)$
defined by $\phi(g) = (gN_1,\ldots,gN_k)$.
Then $\ker\phi = \bigcap_{i=1}^k N_i = N$, which shows that $N$
is a normal subgroup. Moreover, $\im\phi$ is a subgroup of
$\bigoplus_{i=1}^k (G/N_i)$, and by the First Isomorphism Theorem,
$G/N = G/\ker\phi \cong \im\phi$.
\end{proof}

\begin{proposition}
\label{proposition:equalizer cardinality}
Let $G$ and $H$ be finite groups, and let $\Phi \subseteq \Hom(G,H)$.
Let $K \subseteq H$ be the set of $h \in H$ such that
$\bigcap_{\phi \in \Phi}\phi^{-1}(h)$ is nonempty.
Then
\[
\left|\Eq(\Phi)\right| = \left|\bigcap_{\phi \in \Phi}\ker\phi\right| \cdot
\left|K\right|.
\]
\end{proposition}
\begin{proof}
We decompose $\Eq(\Phi)$ into the disjoint union
\[
\Eq(\Phi) = \bigcup_{h \in K}
\left(\bigcap_{\phi \in \Phi}\phi^{-1}(h)\right).
\]
The result then follows from the fact that each
$\bigcap_{\phi \in \Phi}\phi^{-1}(h)$ is a coset of
$\bigcap_{\phi \in \Phi}\ker\phi$, which follows
from Proposition~\ref{proposition:nonempty intersection of preimage is coset}.
\end{proof}

The following proposition is simply the observation that the maximum agreement
between two affine homomorphisms is achievable by two homomorphisms, which
will allow us to reason about homomorphisms rather than affine homomorphisms
in later proofs, without loss of generality.
\begin{proposition}
\label{proposition:lambda hom}
If $G$ and $H$ are finite groups, then there exist $\phi,\psi \in \Hom(G,H)$
such that $\agr(\phi,\psi) = \Lambda_{G,H}$, so if $|\Hom(G,H)| > 1$, then
\[
\Lambda_{G,H} = \max_{\substack{\phi,\psi \in \Hom(G,H)\\ \phi \ne \psi}}
\agr(\phi,\psi)
\]
\end{proposition}
\begin{proof}
Let $\phi',\psi' \in \aHom(G,H)$ such that $\agr(\phi',\psi') = \Lambda_{G,H}$.
By Proposition~\ref{proposition:equalizer is coset}, there exist
$\phi,\psi \in \Hom(G,H)$ such that $|\Eq(\phi,\psi)| = |\Eq(\phi',\psi')|$,
hence $\agr(\phi,\psi) = \agr(\phi',\psi')$.
\end{proof}

Finally, we conclude this section by proving the following.

\begin{proposition}
\label{proposition:index of equalizer to p-group}
Let $G$ be a finite group and let $H$ be a finite $p$-group.
If $\Phi \subseteq \aHom(G,H)$ and $\Eq(\Phi) \ne \emptyset$,
then $[G:\Eq(\Phi)]$ is a power of $p$.
In particular,
\[
\Lambda_{G,H} \le \frac1p.
\]
\end{proposition}
\begin{proof}
By Proposition~\ref{proposition:equalizer is coset}, we may assume
that $\Phi \subseteq \Hom(G,H)$.
It follows from Proposition~\ref{proposition:intersection of normal subgroups}
that $G/\left(\bigcap_{\phi \in \Phi} \ker\phi \right)$
is isomorphic to a subgroup of
$\bigoplus_{\phi \in \Phi} (G/\ker\phi) \cong
\bigoplus_{\phi \in \Phi} \im\phi$.
But the $\im\phi$ are subgroups of $H$, so they are $p$-groups, hence
$\bigoplus_{\phi \in \Phi}\im\phi$ is a $p$-group, and so
$G/\left(\bigcap_{\phi \in \Phi} \ker\phi \right)$
is a $p$-group, i.e.\ 
\[
\frac{|G|}{\left|\bigcap_{\phi \in \Phi} \ker\phi\right|} =
p^k
\]
for some $k$.
By Proposition~\ref{proposition:equalizer cardinality},
there is some integer $m$ such that
\[
\frac{|G|}{|\Eq(\Phi)|} =
\frac{|G|}{\left|\bigcap_{\phi \in \Phi} \ker\phi\right| \cdot m}
= \frac{p^k}{m}.
\]
By Proposition~\ref{proposition:equalizer is coset},
$\Eq(\Phi)$ is a subgroup of $G$, and so by
Lagrange's theorem,
$\frac{p^k}{m} = \frac{|G|}{|\Eq(\Phi)|}$ is an integer, hence
$m$ divides $p^k$, therefore $\frac{p^k}{m}$ is a power of $p$.
\end{proof}

\section{General facts}
\label{section:distance general facts}

In this section, we investigate general properties of $\Lambda_{G,H}$.

\subsection{Subgroups and Quotients}

\begin{proposition}
\label{proposition:range subgroup lambda}
If $G$ and $H$ are finite groups and $K \le H$ is a subgroup,
then
\[
\Lambda_{G,H} \ge \Lambda_{G,K}.
\]
\end{proposition}
\begin{proof}
This follows from the fact that $\aHom(G,K) \subseteq \aHom(G,H)$.
\end{proof}

\begin{proposition}
\label{proposition:quotient lambda}
If $G,H$ are nontrivial finite groups and
$N \triangleleft G$ is a normal subgroup,
then
\[
\Lambda_{G,H} \ge \Lambda_{G/N,H}.
\]
\end{proposition}
\begin{proof}
By Proposition~\ref{proposition:lambda hom}, there exist
$\phi_{G/N},\psi_{G/N} \in \Hom(G/N,H)$ such that
$\agr(\phi_{G/N},\psi_{G/N}) = \Lambda_{G/N,H}$.
Define $\phi,\psi:G \to H$ as follows.
For $x \in G$, define $\phi(x) = \phi_{G/N}(xN)$ and
$\psi(x) = \psi_{G/N}(xN)$. Then $\phi,\psi \in \Hom(G,H)$ since
$\phi$ is the composition of $\phi_{G/N}$ with the natural quotient map
$G \to G/N$, and similarly for $\psi$. It suffices to show that
$\agr(\phi,\psi) = \agr(\phi_{G/N},\psi_{G/N})$, for which it suffices
to show that $|\Eq(\phi,\psi)| = |N| \cdot |\Eq(\phi_{G/N},\psi_{G/N})|$.
This follows from the fact that $\phi$ and $\psi$ are constant on cosets, so
$\Eq(\phi,\psi)$ is a disjoint union of cosets, and the cosets $xN$ on which
$\phi$ and $\psi$ agree are exactly those for which
$\phi_{G/N}(xN) = \psi_{G/N}(xN)$.
\end{proof}

\begin{proposition}
\label{proposition:quotient lambda equality}
If $G,H$ are nontrivial finite groups and $S \le G$ is a subgroup of $G$
such that $|\Hom(S,H)| = 1$, then $\Hom(G,H) \cong \Hom(G/N,H)$, where
$N \trianglelefteq G$ is the smallest normal subgroup of $G$ containing $S$.
In particular,
\[
\Lambda_{G,H} = \Lambda_{G/N,H}.
\]
\end{proposition}
\begin{proof}
Let $\phi \in \Hom(G,H)$. The restriction of $\phi$ to the domain $S$ is a
homomorphism in $\Hom(S,H)$, which is trivial by assumption. This means that
$S \le \ker\phi$. Since $\ker\phi \trianglelefteq G$, by minimality of $N$
it follows that $N \le \ker\phi$. In particular, $\phi=\phi' \circ \pi$ where
$\phi' \in \Hom(G/N,H)$ and $\pi:G \to G/N$ is the natural quotient map.
\end{proof}

\subsection{Zappa-Sz\'ep products}

\begin{proposition}
\label{proposition:product lambda domain}
If $G$ and $H$ are finite groups and $G = G_1 \bowtie G_2$ for some subgroups
$G_1,G_2 \le G$, then
\[
\Lambda_{G,H} \le \max\{\Lambda_{G_1,H},\Lambda_{G_2,H}\}.
\]
\end{proposition}
\begin{proof}
If $|\Hom(G,H)| = 1$, then $\Lambda_{G,H} = 0$ and so the bound is trivial.
Assume $|\Hom(G,H)| > 1$.
By Proposition~\ref{proposition:equalizer is coset}, there exist
$\phi,\psi \in \Hom(G,H)$ such that $\agr(\phi,\psi) = \Lambda_{G,H}$.
First, we introduce some convenient notation.
Denote by $\phi_{G_1}:G_1 \to H$ and $\phi_{G_2}:G_2 \to H$ the restrictions
of $\phi$ to $G_1$ and $G_2$ respectively, and similarly for
$\psi_{G_1}$ and $\psi_{G_2}$. For $y \in G_2$, denote by
$\phi_y:G_1 \to H$ the restriction $\phi_y(x) \triangleq \phi(xy)$.
It is straightforward to verify that
$\phi_{G_i},\psi_{G_i} \in \Hom(G_i,H)$ for $i \in \{1,2\}$ and
$\phi_y,\psi_y \in \aHom(G_1,H)$ for $y \in G_2$.

By averaging, there exists $y \in G_2$ such that
$\agr(\phi_y,\psi_y) \ge \Lambda_{G,H}$. If $\phi_y \ne \psi_y$, then we are
done since
\[
\Lambda_{G,H} \le \agr(\phi_y,\psi_y) \le \Lambda_{G_1,H}.
\]

Otherwise, suppose $\phi_y = \psi_y$. Then $\phi_{G_1} = \psi_{G_1}$, since
for $x \in G_1$,
\[
\phi(x) = \phi_y(x)\phi_y(1_G)^{-1} = \psi_y(x)\psi_y(1_G)^{-1}
= \psi(x).
\]
We claim that
\[
\Eq(\phi,\psi) = G_1 \bowtie\Eq(\phi_{G_2},\psi_{G_2}).
\]
For the forward containment, observe that if $xz \in \Eq(\phi,\psi)$ with
$x \in G_1$ and $z \in G_2$, then
\[
\phi(z) = \phi(x)^{-1}\phi(xz) = \phi_{G_1}(x)^{-1}\phi(xz)
= \psi_{G_1}(x)^{-1}\psi(xz) = \psi(x)^{-1}\psi(xz) = \psi(z)
\]
and so $z \in \Eq(\phi_{G_2},\psi_{G_2})$. Conversely, if
$x \in G_1$ and $z \in \Eq(\phi_{G_2},\psi_{G_2})$, then
\[
\phi(xz) = \phi_{G_1}(x)\phi_{G_2}(z) = \psi_{G_1}(x)\psi_{G_2}(z)
= \psi(xz)
\]
and so $xz \in \Eq(\phi,\psi)$. This completes the proof of our claim.
Moreover, since $\Eq(\phi,\psi) \ne G$, $\Eq(\phi_{G_2},\psi_{G_2}) \ne G_2$,
hence $\phi_{G_2} \ne \psi_{G_2}$. Therefore,
\[
\Lambda_{G,H} =
\frac{|\Eq(\phi,\psi)|}{|G|}
= \frac{|\Eq(\phi_{G_2},\psi_{G_2})|}{|G_2|}
\le \Lambda_{G_2,H}.
\]

\end{proof}

\begin{proposition}
\label{proposition:product coprime lambda}
If $G$ and $H$ are finite groups and $G = G_1 \bowtie G_2$ for some subgroups
$G_1,G_2 \le G$ and $|\Hom(G_2,H)| = 1$, then every
$\phi \in \aHom(G,H)$ is of the form $\phi(xy) = \psi(x)$ for some
$\psi \in \aHom(G_1,H)$ and every $x \in G_1$ and $y \in G_2$.
In particular,
\[
\Lambda_{G,H} \le \Lambda_{G_1,H}
\]
\end{proposition}
\begin{proof}
Suppose $\phi \in \aHom(G,H)$. Then there is some $a \in H$ and some
$\phi_0 \in \Hom(G,H)$ such that $\phi(xy) = a\phi_0(x)\phi_0(y)$ for
every $x \in G_1$ and $y \in G_2$. The restriction of $\phi_0$ to $G_2$
is a homomorphism from $G_2 \to H$, which is trivial by
assumption.
The restriction of $\phi_0$ to $G_1$ is also a homomorphism from $G_1 \to H$.
Thus, $\phi(xy) = \psi(x)$
where $\psi \in \aHom(G_1,H)$ is defined by $\psi(x) = a\phi_0(x)$.
\end{proof}

\subsection{Direct products}

\begin{proposition}
\label{proposition:product lambda}
If $G,H,G_1,G_2,H_1,H_2$ are finite groups, then
\begin{enumerate}
\item%
$\Lambda_{G,H_1 \times H_2} = \max\{\Lambda_{G,H_1},\Lambda_{G,H_2}\}$
\item%
$\Lambda_{G_1 \times G_2,H} = \max\{\Lambda_{G_1,H},\Lambda_{G_1,H}\}$
\end{enumerate}
\end{proposition}
\begin{proof}
\begin{enumerate}
\item%
Since $H_1$ is isomorphic to the subgroup
$H_1 \times \{1_{H_2}\} \le H_1 \times H_2$,
it follows from Proposition~\ref{proposition:range subgroup lambda} that
$\Lambda_{G,H_1 \times H_2} \ge \max\{\Lambda_{G,H_1},\Lambda_{G,H_2}\}$.
For the reverse bound, if $|\Hom(G,H_1 \times H_2)| = 1$, then it is trivial,
so assume $|\Hom(G,H_1 \times H_2)| > 1$.
By Proposition~\ref{proposition:lambda hom}, there exist
$\phi,\psi \in \Hom(G,H)$ with $\agr(\phi,\psi) = \Lambda_{G,H_1 \times H_2}$.
Write $\phi = (\phi_1,\phi_2)$ and $\psi = (\psi_1,\psi_2)$ where
$\phi_i,\psi_i:G \to H_i$ for $i \in \{1,2\}$.
Then $\agr(\phi_1,\psi_1), \agr(\phi_2,\psi_2) \ge \agr(\phi,\psi)
= \Lambda_{G,H_1 \times H_2}$. Moreover, since $\phi \ne \psi$, we have
$\phi_i \ne \psi_i$ for at least one of the $i \in \{1,2\}$. Therefore,
$\Lambda_{G,H_1 \times H_2} \le \agr(\phi_i,\psi_i) \le
\Lambda_{G,H_i} \le \max\{\Lambda_{G,H_1},\Lambda_{G,H_2}\}$.

\item%
Since direct products are Zappa-Sz\'ep products, it follows from
Proposition~\ref{proposition:product lambda domain} that
$\Lambda_{G_1 \times G_2,H} \le \max\{\Lambda_{G_1,H},\Lambda_{G_2,H}\}$.
For the reverse bound, assume without loss of generality that
$\Lambda_{G_1,H} \ge \Lambda_{G_2,H}$.
If $|\Hom(G_1,H)| = 1$, then the bound is trivial, so assume
$|\Hom(G_1,H)| > 1$. By Proposition~\ref{proposition:lambda hom},
there exist $\phi_1,\psi_1 \in \Hom(G_1,H)$ such that
$\agr(\phi,\psi) = \Lambda_{G_1,H}$.
Define $\phi,\psi:G_1 \times G_2 \to H$ by
$\phi(x,y) \triangleq \phi_1(x)$ and $\psi(x,y) \triangleq \psi_1(x)$.
Then $\phi,\psi \in \Hom(G_1 \times G_2,H)$, so
$\Lambda_{G_1 \times G_2,H} \ge \agr(\phi,\psi) = \agr(\phi_1,\psi_1)
= \Lambda_{G_1,H} \ge \max\{\Lambda_{G_1,H},\Lambda_{G_2,H}\}$.
\end{enumerate}
\end{proof}

\subsection{Key facts}

Here we prove some key facts that will help us characterize $\Lambda_{G,H}$
when $G$ is solvable.

\begin{lemma}
\label{lemma:lambda upper bound}
If $G$ and $H$ are finite groups and $p$ is the smallest prime divisor of
$|G|$, then
\[
\Lambda_{G,H} \le \frac1p.
\]
\end{lemma}
\begin{proof}
Suppose $\phi,\psi \in \aHom(G,H)$ are distinct.
By Proposition~\ref{proposition:equalizer is coset}, $\Eq(\phi,\psi)$
is a coset of a subgroup $S$ of $G$, and hence $|\Eq(\phi,\psi)|
= |S|$. By Lagrange's theorem,
$|G|/|S|$ is a divisor of $|G|$, and since $\phi \ne \psi$ it must be
greater than $1$, hence $|G|/|S| \ge p$, so $\agr(\phi,\psi) =
\frac{|\Eq(\phi,\psi)|}{|G|} = \frac{|S|}{|G|} \le \frac1p$.
\end{proof}

\begin{lemma}
\label{lemma:lambda lower bound}
If $G$ has a normal subgroup of index $p$ and $p$ divides $|H|$, then
\[
\Lambda_{G,H} \ge \frac1p.
\]
\end{lemma}
\begin{proof}
Let $N \triangleleft G$ be a normal subgroup of index $p$.
Let $\phi_1: G \to G/N$ be the natural quotient homomorphism.
Since $p$ divides $|H|$, by Cauchy's theorem, there is an element $h \in H$ of
order $p$. The subgroup $\langle h \rangle \le H$ generated by $h$ is 
isomorphic to $\Z_p$, and since $G/N$ has order $p$, it is also isomorphic to
$\Z_p$, hence there is an isomorphism $\phi_2: G/N \to \langle h \rangle$.
Define $\phi:G \to H$ to be the composition $\phi = \phi_2 \circ \phi_1$.
Since $\phi_1,\phi_2$ are homomorphisms, $\phi$ is a homomorphism, and
moreover since $\phi_2$ is an isomorphism, $\ker\phi = \ker\phi_1 = N$.
Therefore, $|\ker\phi| = |N| = |G|/p$.
\end{proof}

\begin{proposition}
\label{proposition:coprime lambda}
If $G$ and $H$ are finite groups and $\gcd(|G|,|H|) = 1$, then
$\aHom(G,H)$ consists of constant functions. In particular,
\[
\Lambda_{G,H} = 0.
\]
\end{proposition}
\begin{proof}
It suffices to show that the only homomorphism $\phi:G \to H$ is the trivial
map $1_H$. If $\phi \in \Hom(G,H)$, then $G/\ker\phi \cong \im\phi$. Moreover,
since $\ker\phi \le G$ and $\im\phi \le H$, $|\im\phi| = |G|/|\ker\phi|$
divides both $|G|$ and $|H|$, hence $|\im\phi| = 1$ and so $\im\phi = \{1_H\}$.
\end{proof}

\section{Nilpotent codomain}
\label{section:distance nilpotent range}

In this section, we prove Theorem~\ref{theorem:main} when $H$ is nilpotent.

We begin by considering the case where $G$ has no normal subgroups of index $p$
for any prime $p$ dividing $\gcd(|G|,|H|)$.
The following fact will be useful.

\begin{proposition}
\label{proposition:solvable maximal normal subgroup index}
If $G$ is a finite solvable group and $N \triangleleft G$ is a maximal normal
subgroup, then $N$ has prime index in $G$.
\end{proposition}

We proceed to prove that $\Lambda_{G,H} = 0$. In fact, we prove it for the case
where $H$ is solvable.

\begin{proposition}
\label{proposition:solvable range no normal subgroup lambda}
Let $G$ and $H$ be finite groups, with $H$ solvable.
If $G$ has no normal subgroup of index $p$ for any prime $p$ dividing
$\gcd(|G|,|H|)$, then $|\Hom(G,H)|=1$ and in particular
\[
\Lambda_{G,H} = 0.
\]
\end{proposition}
\begin{proof}
Suppose $\phi \in \Hom(G,H)$ is nontrivial. Then $\ker\phi \triangleleft G$
is a proper normal subgroup of $G$, and $G/\ker\phi \cong \im\phi$ which is
a subgroup of $H$, and hence solvable. Let $N \triangleleft G$ be a maximal
proper normal subgroup of $G$ containing $\ker\phi$. By the Lattice Theorem,
$N/\ker\phi \triangleleft G/\ker\phi$ is a maximal proper normal subgroup,
so by the Second Isomorphism Theorem and
Proposition~\ref{proposition:solvable maximal normal subgroup index},
$[G:N] = [G/\ker\phi : N/\ker\phi] = p$ for some prime $p$ dividing
$|G/\ker\phi| = |G|/|\ker\phi|$. In particular, $p$ divides $|G|$.
But $p = [G:N]$ divides $[G:\ker\phi] = |\im\phi|$, which divides $|H|$, so $p$
divides $\gcd(|G|,|H|)$. The existence of $N$ contradicts our hypothesis,
so $\phi$ must be trivial.
\end{proof}

This does not hold in general as, for instance, when
$G = H = A_n$ for $n \ge 5$, which is a non-abelian simple group,
$G$ has no normal subgroups of prime index, yet there are certainly
nontrivial homomorphisms $A_n \to A_n$.

Now we proceed to the case where $G$ has a normal subgroup of index $p$ for
some prime $p$ dividing $\gcd(|G|,|H|)$.
We use the well-known fact that finite nilpotent groups are direct
products of their Sylow subgroups~\cite[Ch 6, Theorem 3]{DF04}.

\begin{theorem}
\label{theorem:nilpotent range lambda}
If $G$ is a finite group, $H$ is a finite nilpotent group, and $p$ is the
smallest prime divisor of $\gcd(|G|,|H|)$ such that $G$ has a normal subgroup
of index $p$, then
\[
\Lambda_{G,H} = \frac1p.
\]
\end{theorem}
\begin{proof}
The lower bound follows from Lemma~\ref{lemma:lambda lower bound} so it
suffices to show the upper bound.
Write $H = P_1 \times \cdots \times P_r$ where $P_i$ is the
Sylow $p_i$-subgroup of $H$, and the $p_i$ are distinct.
If $p_i < p$, then $G$ has no normal subgroup of index $p_i$ by assumption,
so by Proposition~\ref{proposition:solvable range no normal subgroup lambda}
it follows that $\Lambda_{G,P_i} = 0$.
On the other hand, if $G$ has a normal subgroup of index $p_i$, then
it follows from Proposition~\ref{proposition:index of equalizer to p-group}
and Lemma~\ref{lemma:lambda lower bound} that
$\Lambda_{G,P_i} = \frac{1}{p_i}$.
Therefore, by Proposition~\ref{proposition:product lambda}, it follows that
$\Lambda_{G,H} = \max_i \Lambda_{G,P_i} = \frac{1}{p}$.
\end{proof}

\section{Solvable domain}
\label{section:distance solvable groups}

In this section, we prove Theorem~\ref{theorem:main} when $G$ is solvable.
As in Section~\ref{section:distance nilpotent range}, we begin by considering
the case where $G$ has no normal subgroups of index $p$ for any prime $p$
dividing $\gcd(|G|,|H|)$.

\begin{proposition}
\label{proposition:no normal subgroup lambda}
Let $G$ be a finite solvable group and let $H$ be any finite group.
If $G$ has no normal subgroup of index $p$ for any prime $p$ dividing
$\gcd(|G|,|H|)$, then $|\Hom(G,H)| = 1$ and in particular
\[
\Lambda_{G,H} = 0.
\]
\end{proposition}
\begin{proof}
Suppose $\phi \in \Hom(G,H)$ is nontrivial.
Then $\ker\phi \triangleleft G$ is a proper normal subgroup of $G$, and
$G/\ker\phi$ is isomorphic to a subgroup of $H$, by the First Isomorphism
Theorem.
In particular, $[G:\ker\phi]$ divides $|H|$.
Let $N \triangleleft G$ be a maximal proper normal subgroup of $G$ containing
$\ker\phi$.
By Proposition~\ref{proposition:solvable maximal normal subgroup index},
$[G:N] = p$ for some prime $p$ dividing $|G|$. But
$p = [G:N]$ divides $[G:\ker\phi]$ which divides $|H|$, so
$p$ divides $\gcd(|G|,|H|)$. By our hypothesis, $N$ cannot exist, so $\phi$
must be trivial.
\end{proof}

We proceed to the case where $G$ has a normal subgroup of index $p$ for
some prime $p$ dividing $\gcd(|G|,|H|)$. Let $p$ be the minimal such prime,
so that we wish to show $\Lambda_{G,H} = \frac1p$.
We first consider the special case where every prime divisor of $|G|$ less than
$p$ also divides $|H|$. In this case, we show that $G$ has no subgroups
of index less than $p$, which yields the upper bound. To show this, we use
the following fact, due to Berkovich, found as an exercise
in~\cite{Isaacs}.

\begin{proposition}[{\cite[Exercise 3B.15]{Isaacs}}]
\label{proposition:solvable minimum index subgroup}
Let $G$ be a finite solvable group. Suppose $H < G$ is a proper subgroup
of $G$ with smallest index. Then $H \triangleleft G$.
\end{proposition}

We now prove the upper bound for the special case.

\begin{lemma}
\label{lemma:solvable dividing lambda}
Suppose $G$ is a finite solvable group, $H$ is any group, and $p$ is the
smallest prime divisor of $\gcd(|G|,|H|)$ such that $G$ has a normal subgroup
of index $p$.
If every prime less than $p$ dividing $|G|$ also divides $|H|$, then
\[
\Lambda_{G,H} \le \frac1p.
\]
\end{lemma}
\begin{proof}
We claim that $G$ has no subgroups of index less than $p$.
Let $S$ be the subgroup with smallest possible index.
By Proposition~\ref{proposition:solvable minimum index subgroup},
$S$ is normal.
Since $S$ is a maximal normal subgroup,
by Proposition~\ref{proposition:solvable maximal normal subgroup index} it
follows that the index $[G:S] = q$ for some prime $q$ dividing $|G|$.
If $q < p$, then our hypotheses imply that $q$ divides $|H|$, so $G$ has a
normal subgroup of prime index less than $p$ dividing $|H|$, contradicting
the minimality of $p$. Thus $[G:S] \ge p$, proving our claim.

By Lemma~\ref{lemma:lambda lower bound}, $|\Hom(G,H)| > 1$, so by
Proposition~\ref{proposition:lambda hom}, there exist homomorphisms
$\phi,\psi \in \Hom(G,H)$ such that $\agr(\phi,\psi) = \Lambda_{G,H}$.
By Proposition~\ref{proposition:equalizer is coset}, $\Eq(\phi,\psi)$ is a
subgroup of $G$, so it follows that
$\Lambda_{G,H} = \agr(\phi,\psi) = 1/[G:\Eq(\phi,\psi)] \le 1/p$.
\end{proof}

We deal with the general case using
the following theorem of Hall~\cite{Hall} characterizing finite solvable groups
as those with Sylow bases.

\begin{theorem}[\cite{Hall}]
\label{theorem:hall solvable}
Let $G$ be a finite group with order prime factorization
$|G| = \prod_{i=1}^m p_i^{e_i}$.
Then $G$ is solvable if and only if it has Sylow $p_i$-subgroups $P_i$ such
that $G = P_1 \bowtie \cdots \bowtie P_m$.
\end{theorem}

We use this decomposition to filter out all the prime divisors of $|G|$ not
dividing $|H|$ to reduce to our special case.

\begin{theorem}
\label{theorem:solvable lambda}
If $G$ is a finite solvable group, $H$ is any group, and $p$ is the smallest
prime divisor of $\gcd(|G|,|H|)$ such that $G$ has a normal subgroup of
index $p$, then
\[
\Lambda_{G,H} = \frac1p.
\]
\end{theorem}
\begin{proof}
The lower bound follows from Lemma~\ref{lemma:lambda lower bound} so it
suffices to show the upper bound.
By Hall's theorem (Theorem~\ref{theorem:hall solvable}), we can write
$G = G_1 \bowtie G_2$ where $\gcd(|G_2|,|H|) = 1$ and every prime dividing
$|G_1|$ divides $|H|$.
By Proposition~\ref{proposition:coprime lambda}, $|\Hom(G_2,H)| = 1$.
Let $N \triangleleft G$ be the smallest normal subgroup
of $G$ containing $G_2$.
By Proposition~\ref{proposition:quotient lambda equality},
$\Lambda_{G,H} = \Lambda_{G/N,H}$, so it suffices to upper bound
$\Lambda_{G/N,H}$.

Since $|G_2|$ divides $|N|$, it holds that
$[G:N]$ divides $[G:G_2] = |G_1|$. In particular, every prime dividing
$|G/N|$ divides $|H|$. Moreover, $G/N$ has no normal subgroups of index
$q < p$, for if it did, it would follow from the Lattice Theorem
that $G$ has a normal subgroup
of index $q$, and moreover $q$ divides $\gcd(|G|,|H|)$, contradicting the
minimality of $p$. Thus, $G/N$ has no normal subgroups of index less than $p$.
Thus, by Lemma~\ref{lemma:solvable dividing lambda}, it follows that
$\Lambda_{G/N,H} \le \frac1p$.
\end{proof}

The formula for $\Lambda_{G,H}$ for solvable $G$ does not extend to
arbitrary finite groups for the obvious reason that $G$ may not have any
normal subgroups of prime index. This holds, for instance, if $G$ is any
non-abelian simple group. One might then hope that the modified statement,
where we drop the requirement that $p$ be prime, holds.
For simple $G$, this formula would be $\Lambda_{G,H} = \frac{1}{|G|}$.
However, the following is a simple (pun intended) counterexample.

Let $G = H = A_5$. Consider the automorphisms which are conjugation by
$(123)$, and its inverse, conjugation by $(132)$. Then these are distinct
homomorphisms, since they disagree on $(12)$ because
$(132)(12)(123) = (13)$ while $(123)(12)(132) = (23)$.
However, they agree on $(45)$ since $(45)$ is a fixed point.
This shows that $\Lambda_{A_5,A_5} \ge \frac{1}{30} > \frac{1}{|G|}$.
In fact, we show in Section~\ref{section:distance simple groups} that
$\Lambda_{A_5,A_5} = \frac{1}{10}$.

\section{Non-abelian simple groups}
\label{section:distance simple groups}

We would like to determine $\Lambda_{G,H}$ for arbitrary finite groups $G$
and~$H$. We propose a two-part strategy for doing this. First, understand
$\Lambda_{G,H}$ for simple groups $G$.
Then, understand how to determine $\Lambda_{G,H}$ for arbitrary $G$ by cleverly
decomposing $G$.
In Section~\ref{section:distance general facts}, we proved some general
facts about $\Lambda_{G,H}$ which could be useful (but far from complete)
for the second part of this program.
In this section, we explore the first part, namely we investigate
$\Lambda_{G,H}$ for non-abelian simple groups~$G$.
A full investigation would entail using the classification of finite simple
groups and considering each family of finite simple groups, which we do not do
in this work.
Instead, we prove some nontrivial lower bounds on $\Lambda_{G,H}$ for general
non-abelian simple $G$.
We then prove some lower and upper bounds on
$\Lambda_{G,G}$ for the specific family $\{A_n\}_{n \ge 5}$ of alternating
groups and pin down $\Lambda_{A_5,A_5} = \frac{1}{10}$ exactly.
We highlight a major difficulty, which is that in the general setting, unlike
in the setting where $G$ is solvable, $\Lambda_{G,H}$ depends on
\emph{how} copies of $G$ are embedded in $H$, not just on the prime divisors
of $|G|$ and $|H|$ and the normal subgroup structure of $G$.

\subsection{Domain and codomain are isomorphic}

If $H$ does not contain a subgroup isomorphic to $G$, then
$\Hom(G,H)$ is trivial. Let us assume that $G = H$. Since $G$ is simple,
$\Hom(G,H) = \Aut(G) \cup \{g \mapsto 1_G\}$.
For $\phi \in \Aut(G)$, $\ker\phi = \{1_G\}$, so clearly $\Lambda_{G,G}
\ge \frac{1}{|G|}$. Can we achieve better agreement?

Better agreement must come from two automorphisms $\phi,\psi \in \Aut(G)$.
Note that $\phi(g) = \psi(g)$ if and only if
$(\phi^{-1} \circ \psi) \in \Aut(G)$ fixes $g$, so we wish to find a
non-identity automorphism
$\phi \in \Aut(G)$ which maximizes $|G^\phi|$, where
\[
G^\phi \triangleq \{g \in G \mid \phi(g) = g\}
\]
is the subset of $G$ fixed by $\phi \in \Aut(G)$.
Observe that the group $\Aut(G)$ naturally acts on the set $G$ via
$\phi \cdot g = \phi(g)$. Let $G/\Aut(G)$ denote the orbits of $G$ under this
group action. By Burnside's lemma,
\[
|G/\Aut(G)| = \frac{1}{|\Aut(G)|}\sum_{\phi \in \Aut(G)} |G^\phi|.
\]
Since $G^{\id} = G$, where $\id \in \Aut(G)$ is the identity automorphism,
\[
|G/\Aut(G)| - \frac{|G|}{|\Aut(G)|} = \frac{1}{|\Aut(G)|}
\sum_{\phi \in \Aut(G), \phi \ne \id} |G^\phi|,
\]
or
\[
\frac{|\Aut(G)|}{|\Aut(G)|-1}
\left(|G/\Aut(G)| - \frac{|G|}{|\Aut(G)|}\right)
= \frac{1}{|\Aut(G)|-1}\sum_{\phi \in \Aut(G), \phi \ne \id}|G^\phi|.
\]
By averaging, this implies that there is some non-identity automorphism
$\phi \in \Aut(G)$ such that
\[
|G^\phi| \ge 
\frac{|\Aut(G)|}{|\Aut(G)|-1}
\left(|G/\Aut(G)| - \frac{|G|}{|\Aut(G)|}\right)
\]
and thus, by dividing by $|G|$, we have
\[
\Lambda_{G,G} \ge
\frac{|\Aut(G)|}{|\Aut(G)|-1}
\left(\frac{|G/\Aut(G)|}{|G|} - \frac{1}{|\Aut(G)|}\right).
\]

\subsection{Alternating groups}

In this section, we prove the following.

\begin{proposition}
\label{proposition:A_n lambda}
For $n \ge 5$,
\[
\frac{2}{n(n-1)} \le \Lambda_{A_n,A_n} \le \frac{1}{n}.
\]
When $n \ne 6$, the upper bound is strict.
\end{proposition}

For $n=5$, the lower bound is tight, that is
$\Lambda_{A_5,A_5} = \frac{|S_3|}{|A_5|} = \frac{1}{10}$.
This is because the only subgroups of $A_n$ larger than $S_3$, up to
isomorphism, are the dihedral group $D_{10}$ of order $10$ generated by
$(1~2~3~4~5)$ and $(2~5)(3~4)$, and $A_4$.
One can check that no conjugation fixes all of $A_4$ nor all of~$D_{10}$.

For the proof of Proposition~\ref{proposition:A_n lambda}, we use the following
fact.
\begin{claim}
\label{claim:maximum subgroup of A_n}
Let $n \ge 3$.
The subgroup $A_{n-1} \le A_n$ is the unique subgroup (up to isomorphism) of
$A_n$ of smallest index. That is, there are no subgroups of $A_n$ with
index less than $n$, and any subgroup of index $n$ is isomorphic to $A_{n-1}$.
\end{claim}
\begin{proof}
First, we show that there are no subgroups of index less than $n$.
Suppose $H \le A_n$ with $m \triangleq [A_n:H] < n$.
The group $A_n$ acts on the left cosets $A_n/H$ by left multiplication,
i.e.\ there is a homomorphism $\rho: A_n \to \Perm(A_n/H) \cong S_m$.
This action is clearly nontrivial, and since $A_n$ is simple, this means
$\rho$ is injective, so $A_n$ embeds into $S_m$. This is impossible since
$n > 2$ implies $|A_n| = \frac{n!}{2} > (n-1)! \ge m! = |S_m|$.

Now, we show uniqueness up to isomorphism.
Let $H \le A_n$ have index $n$. We will show that $H \cong A_{n-1}$.
Again, consider the action $\rho$ as defined above. We established that $A_n$
acts faithfully on $A_n/H$. Observe that $H$ acts on $A_n/H$ by fixing the
coset $H$ and permuting the other $n-1$ cosets. Therefore,
$\rho(H)$ is a subgroup of a copy of $A_{n-1}$ inside $\Perm(A_n/H)$.
Since $\rho$ is injective, $|\rho(H)| = (n-1)!$, and so $\rho(H)$ is actually
isomorphic to $A_{n-1}$. Moreover, $H$ is isomorphic to $\rho(H)$ by the
injectivity of $\rho$, so $H$ is isomorphic to $A_{n-1}$.
\end{proof}

\begin{proof}
[Proof of Proposition~\ref{proposition:A_n lambda}]
For the lower bound, note that there is a
twisted copy of $S_{n-2}$ inside $A_n$, generated by the elements
$(1~2~\cdots~n-2)$ and $(n-1~n)$ when $n$ is even, and by
$(1~2~\cdots~n-1)$ and $(n-1~n)$ when $n$ is odd. In either case,
the automorphism $\phi_\rho : \sigma \mapsto \rho\sigma\rho^{-1}$
with $\rho = (n-1~n)$ fixes this copy of~$S_{n-2}$.

The upper bound follows from the fact that $A_{n-1}$ is the unique subgroup
(up to isomorphism) of $A_n$ of smallest index
(Claim~\ref{claim:maximum subgroup of A_n}).
For $n \ne 6$, every automorphism of $A_n$ is conjugation by some
$\sigma \in S_n$, but no $\sigma \in S_n$ fixes every element of $A_{n-1}$.
\end{proof}

\subsection{Codomain contains copies of domain}

If $H$ contains a copy of $G$, then $\Lambda_{G,H} \ge \Lambda_{G,G}$,
by Proposition~\ref{proposition:range subgroup lambda}.
When $G$ is solvable, it follows from
Proposition~\ref{theorem:solvable lambda} that this is an equality.
One might hope to show that if $G$ is non-abelian simple, then this is actually
an equality, but this is not true.
An easy counterexample is when
$H = A_6$ with subgroups $\Alt({\{1,2,3,4,5\}})$ and $\Alt({\{1,2,3,4,6\}})$
(both isomorphic copies of $A_5$) with $G = \Alt({\{1,2,3,4,5\}})$.
Then $\phi_1:G  \to \Alt({\{1,2,3,4,5\}})$ defined by
$\phi_1(\sigma) = \sigma$ and $\phi_2:G \to \Alt({\{1,2,3,4,6\}})$ defined by
$\phi_2(\sigma) = (5~6)\sigma(5~6)$ agree on $\Alt({\{1,2,3,4,5\}}) \cap
\Alt({\{1,2,3,4,6\}}) = \Alt({\{1,2,3,4\}}) \cong A_4$.
Thus $\Lambda_{A_5,A_6} = \frac{|A_4|}{|A_5|} = \frac{1}{5} >
\frac{1}{10} = \Lambda_{A_5,A_5}$.

\section*{Acknowledgements}

The author thanks Madhu Sudan for inspiring discussions.

\bibliographystyle{alpha}
\bibliography{homomorphism}

\end{document}